\documentclass{article}
\PassOptionsToPackage{numbers}{natbib}
\usepackage[preprint]{neurips_2026}
\usepackage[utf8]{inputenc}
\usepackage[T1]{fontenc}
\usepackage{hyperref}
\usepackage{url}
\usepackage{booktabs}
\usepackage{amsfonts}
\usepackage{amsmath}
\usepackage{amssymb}
\usepackage{nicefrac}
\usepackage{microtype}
\usepackage{xcolor}
\usepackage{graphicx}
\usepackage{subcaption}
\usepackage{multirow}
\usepackage{algorithmic}
\usepackage{algorithm}
\usepackage{enumitem}
\usepackage{tabularx}
\usepackage{colortbl}
\usepackage{amsthm}
\usepackage{thmtools}
\usepackage{thm-restate}
\usepackage[section]{placeins}

\newtheorem{proposition}{Proposition}
\newtheorem{lemma}[proposition]{Lemma}

\newtheorem{remark}{Remark}
\newtheorem{definition}{Definition}

\newcommand{\ours}{UniVoice}

\title{UniVoice: A Unified Model for Speech and Singing Voice Generation}

\author{
  Junjie Zheng\textsuperscript{1}, 
  Huixin Xue\textsuperscript{2}, 
  Shihong Ren\textsuperscript{2}, 
  Chaofan Ding\textsuperscript{1}, 
  Hao Liu\textsuperscript{2}, 
  Zihao Chen\textsuperscript{1} \\
  \textsuperscript{1}Giant Network \\
  \textsuperscript{2}Shanghai Conservatory of Music 
}

\begin{document}
\maketitle

\begin{abstract}
Text-to-speech (TTS) and singing voice synthesis (SVS) both aim to generate human vocal audio from symbolic inputs, but they impose different requirements on the generation process.
Speech generation relies on flexible, language-driven prosody, whereas singing generation requires explicit melody control and accurate rhythmic alignment.
This mismatch makes it challenging to train a single model that can generate both natural speech and controllable singing, since melody-related conditions should strongly constrain singing but should not restrict speech prosody.
We present \ours{}, a unified speech and singing voice generation framework based on conditional flow matching.
Instead of using a single undifferentiated conditioning representation, \ours{} factorizes the condition into content, melody, and timbre, which are encoded by modality-appropriate encoders and consumed by a shared Diffusion Transformer (DiT) backbone.
For singing, the melody condition is represented by MIDI note sequences; for speech, it is replaced with a learned null melody token, allowing the model to infer prosody from linguistic and acoustic context.
This design preserves explicit melody control for singing while avoiding the need to impose melody constraints on speech.
We further analyze the null melody token as an approximation to melody marginalization in the conditional flow.
Trained on 30k hours of speech and 35k hours of singing data, \ours{} achieves a speech PER of 5.26\%, comparable to dedicated TTS systems such as F5-TTS (5.21\%) and CosyVoice3 (5.30\%).
On singing generation, \ours{} achieves a PER of 16.22\%, outperforming the unified baseline Vevo1.5 (24.72\%).
Ablation studies show that factorized conditioning, task modulation, and learned null melody conditioning are important for unified generation.
We also introduce \textsc{UniSinging-Eval}, a benchmark covering 12 musical styles for evaluating unified speech and singing generation.
We will release the inference code, model checkpoints, and the \textsc{UniSinging-Eval} test set.
Audio samples are available at: \url{https://nips-unvoice.netlify.app/}.
\end{abstract}


\section{Introduction}
\label{sec:intro}

Text-to-speech (TTS) and singing voice synthesis (SVS) both aim to generate human vocal audio from symbolic inputs, but they are typically developed as separate systems.
Modern TTS models are designed to produce intelligible and natural speech from text, where prosody is largely inferred from linguistic context.
In contrast, SVS models require explicit musical control, often in the form of MIDI notes or pitch contours, to specify melody and rhythm~\cite{promptsinger,vevo2}.
This separation becomes limiting in applications that require a consistent vocal identity across speaking, singing, and intermediate styles such as rap or musical theatre~\cite{valle2,cosyvoice}.

A unified model for speech and singing is appealing, but it faces a fundamental mismatch in conditioning.
For singing, melody should act as a strong constraint: the generated voice must follow the given pitch and timing.
For speech, however, imposing an explicit melody is usually unnecessary and can restrict natural, context-dependent prosody.
Na\"ively training a single model on mixed speech and singing data therefore forces the model to assign incompatible meanings to melody-related conditions.
This can lead to gradient interference and degraded generation quality, as the two modalities place different demands on the conditioning space (Figure~\ref{fig:overview}).

We propose \ours{}, a unified framework for speech and singing voice generation.
The key idea is to share the generative model while factorizing the conditioning signals that control it.
Specifically, \ours{} decomposes the condition into content, melody, and timbre.
Content is represented by phoneme sequences, timbre is provided through an audio prompt, and melody is represented by MIDI notes when generating singing.
For speech, where explicit melody control is not needed, the melody condition is replaced with a learned null token.
This allows the model to preserve melody control for singing while letting speech prosody be inferred from linguistic and acoustic context.

The generator is built on Conditional Flow Matching (CFM)~\cite{flowmatching} with a Diffusion Transformer (DiT) backbone~\cite{dit}, adapted from scalable video generation architectures~\cite{wandit}.
Speech and singing are treated as two conditional modes of the same continuous flow, and a lightweight task token modulates the shared backbone without introducing separate task-specific heads.
We further show that the learned null melody token can be interpreted as marginalizing over the melody variable for speech generation, providing a principled way to avoid imposing singing-specific controls on speech.

We train \ours{} on 30k hours of speech and 35k hours of singing data.
With 0.3B parameters, \ours{} achieves a speech PER of 5.26\%, comparable to dedicated TTS systems such as F5-TTS (5.21\%).
On singing generation, it substantially outperforms the unified baseline Vevo2, reducing PER from 24.72\% to 16.22\%.
Ablation studies show that factorized conditioning is essential for both modalities, and representation analyses indicate that the shared backbone preserves modality-specific structure while maintaining a common vocal generation space. \\
To facilitate reproducible research, we will release the inference code, model checkpoints, and the \textsc{UniSinging-Eval} test set.

Our contributions are summarized as follows:
\begin{enumerate}[leftmargin=*,itemsep=1pt]
    \item We introduce a unified flow-matching architecture for speech and singing voice generation, using a single DiT backbone without separate task-specific heads.
    \item We propose a factorized conditioning scheme that separates content, melody, and timbre, reducing the conflict between speech prosody modeling and singing melody control.
    \item We introduce learned null melody conditioning for speech and analyze its connection to melody marginalization in conditional flow matching.
    \item We present \textsc{UniSinging-Eval}, a benchmark for evaluating unified voice generation across 12 musical styles with phrase-level timestamps and style metadata.
\end{enumerate}

\begin{figure*}[t]
    \centering
    \includegraphics[width=\textwidth]{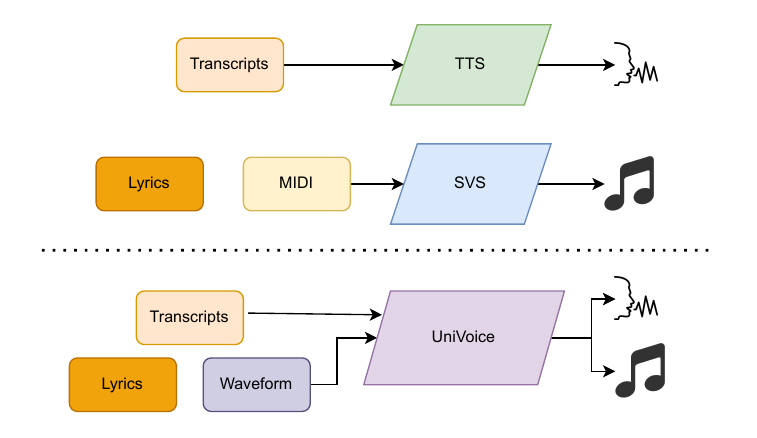}
    \caption{Comparison of speech and singing voice generation paradigms.
Speech synthesis (TTS) infers prosody implicitly from text, while singing voice synthesis (SVS) requires explicit melody control through MIDI notes.
Naïvely combining the two leads to conflicting conditioning semantics.
UniVoice resolves this conflict through factorized conditioning and a learned null melody token for speech.}
    \label{fig:overview}
\end{figure*}

\section{Related Work}
\label{sec:related}

\paragraph{Text-to-Speech Synthesis.}
Modern TTS has evolved from codec language models such as VALL-E~\cite{valle,valle2} to flow-based architectures including Voicebox~\cite{voicebox}, F5-TTS~\cite{f5tts}, NaturalSpeech~2/3~\cite{naturalspeech2,naturalspeech3}, and CosyVoice~\cite{cosyvoice}, achieving strong zero-shot and multilingual synthesis performance. However, these systems do not support singing because they lack explicit melody control and rhythmic alignment mechanisms required for SVS.

\paragraph{Singing Voice Synthesis.}
Modern SVS systems leverage diffusion and transformer-based architectures for expressive melody-conditioned synthesis~\cite{promptsinger,makesinger,ditsinger}. While dedicated singing models achieve strong quality, they generally rely on precise MIDI or $F_0$ supervision and cannot generate natural speech.

\paragraph{Unified Speech and Singing.}
Vevo2~\cite{vevo2} explored unified speech and singing generation with shared autoregressive modeling and explicit melody conditioning. However, rigid $F_0$ constraints introduce a trade-off between natural speech prosody and singing controllability. Our factorized conditioning addresses this conflict through independently controlled content, melody, and timbre representations with null-token marginalization.

\paragraph{Flow Matching and DiT.}
CFM~\cite{flowmatching} enables stable and efficient continuous generative modeling via ODE-based flow learning. DiT~\cite{dit} and Wan-DiT~\cite{wandit} demonstrate that transformer-based diffusion architectures with AdaLN, RoPE~\cite{rope}, and FlashAttention-2~\cite{flash2} effectively model long-range dependencies. We adapt this design for unified audio generation.

\section{Method}
\label{sec:method}

\subsection{Preliminaries: Conditional Flow Matching}
\label{sec:cfm}

\ours{} builds upon Conditional Flow Matching (CFM)~\cite{flowmatching}, which models continuous probability paths between Gaussian noise $p_0=\mathcal{N}(0,I)$ and target audio latent features $x_1 \sim p_{\text{data}}$ extracted by a VAE encoder. 
Under the optimal transport (OT) path, the interpolant is:
\begin{equation}
    \phi_t(x_0, x_1) = (1-t)x_0 + tx_1,
\end{equation}
where $x_0 \sim \mathcal{N}(0,I)$. 
The vector field network $v_\theta(x,t)$ is trained with:
\begin{equation}
    \mathcal{L}_{\text{CFM}} =
    \mathbb{E}_{t,x_0,x_1}
    \left[
    \left\|
    v_\theta(\phi_t(x_0,x_1), t)
    - (x_1-x_0)
    \right\|^2
    \right].
\end{equation}
At inference, samples are generated by solving $\frac{dx}{dt}=v_\theta(x,t)$.

For conditional generation, we extend the vector field to $v_\theta(x_t,t,\mathbf{c})$. Since the OT interpolant is independent of $\mathbf{c}$, speech and singing share the same transport geometry while differing only in conditioning directions, enabling positive transfer through the shared backbone.

Compared to score-based diffusion and codec language models, CFM avoids noise scheduling, supports efficient ODE sampling (32 steps), and naturally models both speech and singing in a continuous latent space without codec quantization artifacts.

\begin{figure}[t]
    \centering
    \includegraphics[width=0.95\columnwidth]{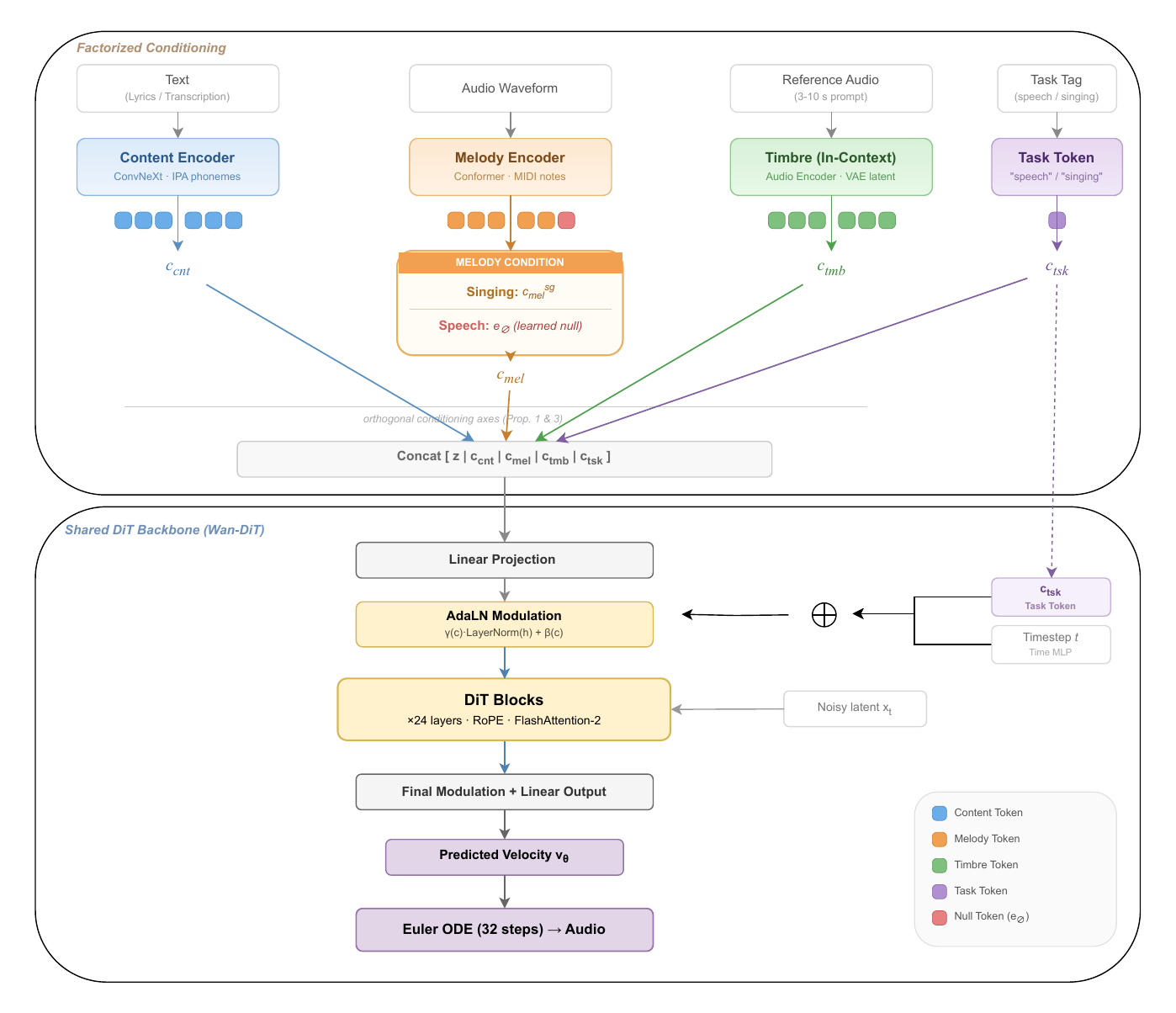}
    \caption{Architecture of \ours{}. The factorized conditioning region (top) processes content, melody, and timbre through independent encoders.  For speech, the melody encoder is replaced by a learned null token $\mathbf{e}_\varnothing$. All conditions are concatenated and fed to the shared DiT backbone (bottom), which is modulated by the task token via AdaLN.}
    \label{fig:arch}
\end{figure}

\subsection{Architecture}
\label{sec:arch}

The backbone of \ours{} is a scalable Diffusion Transformer (DiT) adapted from the Wan-DiT video generation architecture~\cite{wandit}.
This choice is deliberate: in emerging audio-visual joint generation pipelines, the video branch is typically built on the Wan-DiT family, and audio DiT modules are trained from scratch with the same architecture to facilitate cross-modal fusion.
By adopting an identical backbone, \ours{}'s audio DiT can be directly integrated into such pipelines without architectural adaptation.

\paragraph{Model Configuration.}
The DiT backbone consists of $L{=}24$ transformer layers with hidden dimension $D{=}1024$, $H{=}16$ attention heads, and $4{\times}$ FFN expansion, yielding approximately 0.3B parameters.
We use FlashAttention-2~\cite{flash2} and RoPE~\cite{rope} for efficient long-context modeling.
The model operates on Song Bloom VAE~\cite{songbloomvae} latent features at 25\,Hz with dimension 48.

\paragraph{Feature Extraction and Fusion.}
The input to the DiT consists of noisy latent features $x_t$ concatenated channel-wise with the conditioning signals $[x_t \,|\, c_{\text{cnt}} \,|\, c_{\text{mel}} \,|\, c_{\text{tmb}}] \in \mathbb{R}^{T \times (48 + 3D)}$, projected to the hidden dimension $D$ via a linear layer.
This concatenation-based fusion is simple but effective: it allows the DiT to learn cross-modal interactions through self-attention rather than requiring explicit cross-attention modules.
To differentiate between speech and singing tasks, we employ a \emph{Task Mapper}---a lightweight projection that embeds a learnable task token into the input.
This token is processed through AdaLN modulation, allowing the model to adapt its generation dynamics to each modality without introducing task-specific parameters in the main computation path.

\paragraph{Adaptive Modulation.}
The task embedding modulates the feature distribution at every layer via adaptive layer normalization (AdaLN):
\begin{equation}
    \text{AdaLN}(h, c) = \gamma(c) \cdot \frac{h - \mu(h)}{\sigma(h)} + \beta(c),
    \label{eq:adaln}
\end{equation}
where $c = c_{\text{time}} + c_{\text{tsk}}$ is the sum of the timestep embedding and task token embedding, and $\gamma(c), \beta(c)$ are predicted scale and shift parameters.

\subsection{Factorized Conditioning}
\label{sec:factorized}

We introduce a factorized conditioning scheme to resolve the representation conflict between speech and singing.
We decompose the conditioning signal into four components:
\begin{equation}
    \mathbf{c} = (c_{\text{cnt}}, c_{\text{mel}}, c_{\text{tmb}}, c_{\text{tsk}}),
    \label{eq:factorize}
\end{equation}
where $c_{\text{cnt}}$ captures linguistic content, $c_{\text{mel}}$ captures melodic structure (MIDI note sequences), $c_{\text{tmb}}$ captures speaker timbre (from audio prompts), and $c_{\text{tsk}} \in \{\text{sp}, \text{sg}\}$ indicates the generation modality.

\paragraph{Why factorize?}
The representation conflict arises because a monolithic encoder creates \emph{negative gradient correlation} between speech and singing on melody-related parameters (Definition~\ref{def:conflict}, Appendix~\ref{app:proofs}).
Speech pushes melody parameters toward a constant function, while singing pushes them toward faithfully encoding MIDI note sequences.
Factorizing into separate encoders eliminates this: the melody encoder receives gradients only from singing, and the null token only from speech.


\begin{proposition}[Factorization reduces conflict]
\label{prop:factorize}
Suppose melody is approximately conditionally independent of speech audio given content and timbre:
$x \perp\!\!\!\perp c_{\text{mel}} \mid c_{\text{cnt}}, c_{\text{tmb}}, c_{\text{tsk}} = \text{sp}$.
Then factorized conditioning removes the negative cross-task interaction on the melody-related subspace.
In particular, the magnitude of the eliminated conflict can be written as
\begin{equation}
    \Delta
    =
    -\alpha(1-\alpha)
    \left\langle
    \nabla_{\theta_{\text{mel}}}\mathcal{L}_{\text{sp}}^{\text{mono}},
    \nabla_{\theta_{\text{mel}}}\mathcal{L}_{\text{sg}}^{\text{mono}}
    \right\rangle
    \geq 0,
    \label{eq:bound}
\end{equation}
whenever the monolithic model exhibits negative gradient correlation between speech and singing on melody-related parameters.
\end{proposition}

\noindent Proof in Appendix~\ref{app:proof_factorize}. The four axes are instantiated as follows:

\paragraph{Content Conditioning.}
Phoneme sequences are tokenized using IPA and processed by a ConvNeXt-based encoder~\cite{f5tts}: $c_{\text{cnt}} = \text{ConvNeXt}(\mathbf{p}) \in \mathbb{R}^{N \times D}$.

\paragraph{Melody Conditioning.}
Singing requires explicit melody specification, whereas speech relies on implicit prosody.
We design a dual-path strategy:
\begin{itemize}[leftmargin=*,itemsep=1pt]
    \item \textbf{For singing:} A Conformer~\cite{conformer} encodes MIDI note sequences into a dense melodic representation: $c_{\text{mel}}^{\text{sg}} = \text{Conformer}(\text{MIDI}) \in \mathbb{R}^{T \times D}$.  Each MIDI note provides discrete pitch and duration information, which the Conformer converts into frame-level melodic embeddings aligned with the content sequence.
    \item \textbf{For speech:} A \emph{learned null token} $c_{\text{mel}}^{\text{sp}} = \mathbf{e}_\varnothing \in \mathbb{R}^{D}$, broadcast across $T$, signals the model to infer prosody from context.  The theoretical justification is in Section~\ref{sec:theory}.
\end{itemize}

\paragraph{Timbre Conditioning via In-Context Learning.}
Following E2~TTS~\cite{e2tts}, a 3--10\,s reference audio prompt provides zero-shot voice cloning.
The reference audio is encoded by the same Song Bloom VAE encoder used for target audio, producing a latent $c_{\text{tmb}} \in \mathbb{R}^{T' \times D}$ that is concatenated with the target sequence along the time axis.
The DiT backbone attends to this prompt through its self-attention mechanism to infer speaker identity---including vocal quality, pitch range, and stylistic characteristics---without any fine-tuning or speaker embeddings.
This in-context approach is particularly advantageous for unified generation, as the same reference audio can condition both speech and singing outputs, ensuring consistent timbre across modalities.

\paragraph{Task Token.}
A learnable embedding $c_{\text{tsk}} \in \mathbb{R}^{D}$ is injected into AdaLN (Eq.~\ref{eq:adaln}), providing a discrete mode switch and disambiguating the null melody token semantics (``melody is irrelevant for speech'' vs.\ ``melody is informative for singing'').

\subsection{Theoretical Analysis of Null-Token Marginalization}
\label{sec:theory}

A natural question arises: \emph{why does replacing melody conditioning with a null token work for speech, rather than degrading quality?}
We provide a formal justification; full proofs are in Appendix~\ref{app:proofs}.

\paragraph{The core problem.}
For speech generation, the ideal vector field should marginalize over all possible melody realizations:
$v^{*}_{\text{sp}} = \int v^{*}(\cdot, c_{\text{mel}}) \, p(c_{\text{mel}}) \, dc_{\text{mel}}$.
This integral is intractable.
The learned null token $\mathbf{e}_\varnothing$ provides an efficient amortized approximation:

\begin{proposition}[Null-token optimality]
\label{prop:null}
Under the joint training loss, the speech component $\mathcal{L}_{\text{sp}}$ depends on $\mathbf{e}_\varnothing$ but not on the melody encoder $g_\phi$, and vice versa for $\mathcal{L}_{\text{sg}}$.
At convergence, the optimal null token satisfies:
\begin{equation}
    v_\theta(x_t, t, c_{\text{cnt}}, \mathbf{e}_\varnothing^{*}, c_{\text{tmb}}, \text{sp}) \approx v^{*}_{\text{sp}}(x_t, t, c_{\text{cnt}}, c_{\text{tmb}}),
    \label{eq:null_approx}
\end{equation}
i.e., the velocity field with the null token recovers the Bayes-optimal melody-marginalized velocity field.
\end{proposition}
\noindent Intuitively, since $\mathbf{e}_\varnothing$ and $g_\phi$ receive non-overlapping gradients, the null token learns the optimal ``melody-absent'' representation without interference.
A learned null token is strictly more expressive than a fixed zero vector (Remark~\ref{rem:zero}, Appendix~\ref{app:proofs}).
Proof in Appendix~\ref{app:proof_null}.

\paragraph{Connection to classifier-free guidance.}
Our approach can be viewed as \emph{axis-selective} classifier-free guidance~\cite{cfg}: instead of dropping all conditions, we selectively marginalize only the melody axis.
At inference, this enables per-axis guidance:
\begin{equation}
    \tilde{v}(x_t, t, \mathbf{c}) = v_\theta(x_t, t, \mathbf{c}) + w_j \cdot \bigl[ v_\theta(x_t, t, \mathbf{c}) - v_\theta(x_t, t, c_1, \ldots, \mathbf{e}_\varnothing^{(j)}, \ldots, c_K) \bigr],
    \label{eq:axis_cfg}
\end{equation}
where only axis $j$ is replaced.



\subsection{Training Strategy}
\label{sec:training}

\paragraph{Data.}
We train \ours{} on a large-scale mixed dataset comprising approximately 65k hours of audio: 35k hours of singing recordings and 30k hours of speech.
The singing corpus is derived from in-the-wild songs and studio-quality dry recordings; for the former, we isolate vocal tracks using source separation.
Lyrics are aligned via ASR cross-validation (when SRT files are available) or VAD + ASR, followed by a rule-based cleaning pipeline to filter low-quality transcriptions.
Phoneme tokens are pre-extracted using a G2P module, vocal latents are computed with a pre-trained VAE, and MIDI sequences are generated by a pre-trained MIDI extraction model to provide melody conditioning.
The speech data is sourced from the Emilia dataset~\cite{emilia}; we sample a balanced subset of $\sim$30k hours (15k Chinese, 15k English) to match the singing data distribution.
We use a unified IPA phoneme system for both speech and singing, enabling shared content representations across modalities.
For speech, the melody condition is replaced by the null token $\mathbf{e}_\varnothing$, encouraging the model to infer prosody from linguistic context rather than explicit melody.

\paragraph{Training Procedure.}
During training, we randomly drop text, melody, and audio-prompt conditions independently with probabilities $p_{\text{text}} = p_{\text{melody}} = p_{\text{timbre}} = 0.1$ to enable axis-selective classifier-free guidance (Eq.~\ref{eq:axis_cfg}) at inference time.
Speech and singing samples are mixed at their natural ratio ($\sim$46\%/54\%) within each batch.
The batch size is 32{,}000 frames per GPU on 8$\times$A800 (80\,GB) GPUs.
Training spans 100 epochs with the AdamW optimizer, a peak learning rate of $7.5 \times 10^{-5}$, and a linear warmup of 2{,}000 steps followed by cosine decay.
Additional architecture and training details are provided in Appendix~\ref{app:arch}.

\paragraph{Inference.}
We use an Euler ODE solver with 32 steps and exponential moving average (EMA) weights.
Classifier-free guidance is applied with axis-specific scales: text guidance strength $w_{\text{txt}}{=}5$ for strong linguistic adherence, and audio/melody guidance strength $w_{\text{aud}}{=}w_{\text{mel}}{=}1$ for faithful timbre and melody reproduction.
The generated mel-spectrogram latent is decoded by the Song Bloom VAE decoder to produce the final waveform.


\section{UniSinging-Eval Benchmark}
\label{sec:dataset}

To support reproducible evaluation of unified voice generation, we construct \textsc{UniSinging-Eval}, a curated benchmark that addresses key gaps in existing evaluation protocols.

\paragraph{Design Principles.}
Existing SVS benchmarks such as those used in PromptSinger~\cite{promptsinger} and MakeSinger~\cite{makesinger} focus on narrow musical styles (typically Pop and Ballad only) and do not test the ability to transition between speech and singing with consistent speaker identity.
\textsc{UniSinging-Eval} addresses three evaluation dimensions: (1)~singing quality across diverse musical styles with varying rhythmic and melodic complexity, (2)~speech quality on standard benchmarks to ensure no regression, and (3)~cross-modal consistency, testing whether the same audio prompt produces coherent speaker identity across speech and singing outputs.

\paragraph{Dataset Construction.}

We curate 60 songs spanning 12 diverse musical styles: Pop, Hip-Hop/Rap, Pop Electronic, Rock, Country, City Pop, R\&B, Afrobeats, Bossa Nova, Jazz, Funk, and Neo-Soul.
Each song is segmented into phrase-level clips with aligned MIDI note annotations and IPA-transcribed lyrics, yielding 900 samples totaling 2 hours of audio.
To test robustness to lyrical variation---a key requirement for practical singing voice synthesis---the dataset is organized into three difficulty levels:
\begin{itemize}[leftmargin=*,itemsep=1pt]
    \item \textbf{Level~1 (Minor):} Original lyrics preserved; tests baseline synthesis quality.
    \item \textbf{Level~2 (Major):} Lyrics extensively rewritten with complex vocabulary and changed themes; tests phoneme--melody alignment under challenging text.
    \item \textbf{Level~3 (Elastic):} Arbitrary themes with $\pm5$ word count variation and filler/transition words; tests model robustness to duration mismatch between lyrics and melody.
\end{itemize}
Full dataset statistics and rebuild constraints are provided in Appendix~\ref{app:dataset}.
For speech evaluation, we use the Emilia~\cite{emilia} test set, following the standard cross-sentence protocol.

\section{Experiments}
\label{sec:experiments}

\subsection{Experimental Setup}
\label{sec:setup}

\paragraph{Baselines.}
We compare against the following systems, selected to cover both dedicated and unified approaches:
\begin{itemize}[leftmargin=*,itemsep=1pt]
    \item \textbf{F5-TTS}~\cite{f5tts}: A state-of-the-art speech-only TTS model (0.3B parameters) based on flow matching with a ConvNeXt text encoder, trained on 100k hours of speech data.
    \item \textbf{CosyVoice3}~\cite{cosyvoice}: A large-scale flow-matching TTS model (0.5B) with multilingual support, trained on 1M hours of speech.
    \item \textbf{Vevo1.5}~\cite{vevo2}: A unified speech and singing model (1B) using autoregressive codec prediction with explicit $F_0$ conditioning, trained on 101k hours of speech and 7k hours of singing.
    \item \textbf{Soul-X-Singer}: A dedicated singing model (0.7B) with multi-scale diffusion, trained on 42k hours of singing data.
\end{itemize}

\paragraph{Evaluation Protocol.}
We evaluate on a cross-sentence generation task: given a reference text and a short audio prompt (3--10\,s), the model synthesizes audio matching the target text while preserving the prompt speaker's voice.
For singing, the model additionally receives MIDI note sequences as melody conditioning.
All models generate 10\,s clips at 48\,kHz and are evaluated on the same test splits.

\paragraph{Metrics.}
We report four metrics: \textbf{PER} (phoneme error rate via FireRedASR~\cite{fireredasr}), which measures intelligibility; \textbf{SIM} (speaker similarity via WavLM cosine distance), which measures voice cloning fidelity; \textbf{S-MOS} (similarity mean opinion score), a human evaluation of perceived speaker match; and \textbf{N-MOS} (naturalness mean opinion score), a human evaluation of overall audio quality.
Human evaluations are conducted with 20 listeners per sample, each rating on a 1--5 Likert scale.

\subsection{Main Results}
\label{sec:main_results}

Table~\ref{tab:main} presents the comprehensive comparison.

\paragraph{Speech Generation.}
\ours{} achieves a PER of 5.26\%, comparable to dedicated speech systems such as F5-TTS (5.21\%) and CosyVoice3 (5.30\%), despite using substantially less speech data and supporting unified speech--singing generation within a single 0.3B model.
It also achieves the highest speech naturalness score (S-MOS 3.76).
The SIM gap relative to speech-only systems reflects the trade-off introduced by unified training and reduced speech data scale.
Compared to Vevo1.5, \ours{} significantly improves speech quality (PER 5.26\% vs.\ 14.10\%) with fewer parameters.

\paragraph{Singing Voice Generation.}
\ours{} achieves a singing PER of 16.22\% and SIM of 35.70\%, substantially outperforming Vevo1.5 (PER 45.07\%) and the dedicated Soul-X-Singer baseline (PER 26.22\%).
It also achieves the highest singing naturalness score (N-MOS 3.25).
These results suggest that factorized conditioning improves phoneme--melody alignment while enabling positive transfer between speech and singing representations.

\paragraph{Data and Parameter Efficiency.}
Despite using only 0.3B parameters and 65k training hours, \ours{} outperforms larger unified and dedicated singing systems.
Compared to Vevo1.5 (1B, 108k hours) and Soul-X-Singer (0.7B), \ours{} achieves better singing quality while remaining competitive on speech generation.
This efficiency likely comes from reduced gradient interference through factorized conditioning and shared representation learning between speech and singing.


\begin{table}[t]
\centering
\caption{Performance comparison on speech and singing tasks. PER (\%)$\downarrow$, SIM (\%)$\uparrow$, S-MOS$\uparrow$, N-MOS$\uparrow$. Best results are in \textbf{bold}, and second-best results are \underline{underlined}.}
\label{tab:main}
\small
\setlength{\tabcolsep}{3.5pt}
\resizebox{\linewidth}{!}{
\begin{tabular}{lcc|cccc|cccc}
\toprule
\multirow{2}{*}{\textbf{Model}} &
\multirow{2}{*}{\textbf{Size}} &
\multirow{2}{*}{\textbf{Data}} &
\multicolumn{4}{c|}{\textbf{Speech}} &
\multicolumn{4}{c}{\textbf{Singing}} \\
& & &
PER & SIM & S-MOS & N-MOS &
PER & SIM & S-MOS & N-MOS \\
\midrule
F5-TTS        & 0.3B & 100k+0   & \textbf{5.21} & \underline{72.73} & \textbf{3.85} & \textbf{3.41} & -- & -- & -- & -- \\
CosyVoice3    & 0.5B & 1000k+0  & 5.30 & \textbf{74.94} & 3.75 & \underline{3.19} & -- & -- & -- & -- \\
Vevo1.5       & 1B   & 101k+7k  & 14.10 & 59.07 & 2.85 & 2.77 & 45.07 & \underline{36.33} & \underline{2.62} & 2.69 \\
Soul-X-Singer & 0.7B & 0+42k    & -- & -- & -- & -- & \underline{26.22} & \textbf{41.71} & \textbf{3.19} & \underline{3.24} \\
\midrule
\ours{}       & 0.3B & 30k+35k  & \underline{5.26} & 67.42 & \underline{3.76} & 3.07 & \textbf{16.22} & 35.70 & \textbf{3.19} & \textbf{3.25} \\
\bottomrule
\end{tabular}
}
\end{table}

\subsection{Ablation Studies}
\label{sec:ablation}

To validate the contribution of each component, we conduct extensive ablation studies.
All ablation models are trained with the same data and hyperparameters unless otherwise noted.




\begin{table}[t]
\centering
\caption{Ablation study on conditioning components. All variants use the same 0.3B DiT backbone. PER~(\%)$\downarrow$, SIM~(\%)$\uparrow$, S-MOS$\uparrow$, N-MOS$\uparrow$.}
\label{tab:ablation_cond}
\small
\begin{tabular}{l|cccc|cccc}
\toprule
\multirow{2}{*}{\textbf{Configuration}} & \multicolumn{4}{c|}{\textbf{Speech}} & \multicolumn{4}{c}{\textbf{Singing}} \\
 & PER & SIM & S-MOS & N-MOS & PER & SIM & S-MOS & N-MOS \\
\midrule
\ours{} (Full)
    & \textbf{5.26}  & \textbf{67.42} & \textbf{3.74} & \textbf{3.07}
    & \textbf{16.22} & \textbf{35.70} & \textbf{3.19} & \textbf{3.23} \\
\midrule
w/o factorized cond.
    & 12.31 & 65.84 & 2.89 & 2.31
    & 23.45 & 31.08 & 2.23 & 2.47 \\
w/o task token
    & 8.34 & 66.35 & 3.27 & 2.65
    & 19.92 & 33.21 & 3.01 & 3.02 \\
w/o null melody token
    & 7.86 & 67.18 & 3.62 & 3.00
    & 18.64 & 34.30 & 3.06 & 3.11 \\
w/o melody encoder
    & 6.27 & 67.03 & 3.58 & 2.95
    & 23.21 & 34.19 & 2.03 & 2.42 \\
\bottomrule
\end{tabular}
\end{table}

\paragraph{Factorized conditioning is critical.}
Removing factorized conditioning causes the largest degradation, increasing speech PER from 5.26\% to 12.31\% and singing PER from 16.22\% to 23.45\%. This confirms that naive joint training cannot effectively balance speech prosody and melody-conditioned singing generation.

\paragraph{Task tokens improve modality adaptation.}
Removing the task token degrades both speech and singing performance, indicating that modality-specific AdaLN modulation is important for unified generation.

\paragraph{Learned null tokens outperform fixed zeros.}
Replacing the learned null melody token with a fixed zero vector degrades both tasks, suggesting that the learned embedding better captures the absence of melody information.

\paragraph{Melody conditioning is essential for singing.}
Removing the melody encoder has little effect on speech but significantly degrades singing quality, confirming that explicit melody control is indispensable for SVS.

\section{Broader Impact and Limitations}
\label{sec:impact}

Unified voice generation can improve accessibility, music creation, and educational applications by enabling personalized speech and singing synthesis. 
For example, it may help users with vocal impairments create expressive vocal content, support rapid prototyping for musicians, and provide personalized singing demonstrations for music learning.
However, zero-shot voice cloning also raises risks related to deepfakes, copyright infringement, and labor displacement.
These risks are particularly relevant when models can preserve speaker identity across both speech and singing.
We recommend watermarking, provenance tracking, and consent-based speaker enrollment for responsible deployment.

Unified training still shows a SIM gap compared to speech-only systems, suggesting a trade-off between cross-modal sharing and speaker similarity.
Evaluation is limited to English and Chinese, leaving multilingual and low-resource generalization underexplored.
Mixed vocal styles such as rap and spoken singing are not explicitly modeled, and 32-step ODE sampling introduces higher latency than few-step generation methods.

\section{Conclusion}
\label{sec:conclusion}

We presented \ours{}, a unified flow-matching framework for speech and singing generation based on factorized conditioning and a shared DiT backbone.
By decomposing content, melody, and timbre, \ours{} enables a single 0.3B model to generate natural speech and expressive singing without separate task heads, while a learned null melody token avoids imposing explicit melody constraints on speech.

Experiments show competitive speech quality with dedicated TTS systems and substantially outperform prior unified baselines on singing generation.
Ablations confirm the importance of factorized conditioning and null-token melody modeling, and representation analyses suggest that the shared backbone preserves modality-specific structure.
We introduce \textsc{UniSinging-Eval} for unified speech and singing evaluation.
Future work includes scaling model capacity, extending multilingual support, and integrating unified audio generation into multimodal audio-visual pipelines.

\appendix

\section{Extended Proofs and Theoretical Details}
\label{app:proofs}

This appendix provides detailed proofs and additional theoretical analysis for the results presented in Sections~\ref{sec:factorized} and \ref{sec:theory}.

\subsection{Representation Conflict: Formal Definition}
\label{app:def_conflict}

\begin{definition}[Representation conflict]
\label{def:conflict}
Let $\mathcal{M}_{\text{sp}}$ and $\mathcal{M}_{\text{sg}}$ denote the manifolds of speech and singing audio in latent space.
A \emph{representation conflict} occurs when a single conditioning encoder $f_\theta$ must simultaneously satisfy:
\begin{align}
    I\bigl(f_\theta(\mathbf{c});\, x \mid x \in \mathcal{M}_{\text{sp}}\bigr) &\geq H(x \mid x \in \mathcal{M}_{\text{sp}}) - \epsilon, \\
    I\bigl(f_\theta(\mathbf{c});\, x \mid x \in \mathcal{M}_{\text{sg}}\bigr) &\geq H(x \mid x \in \mathcal{M}_{\text{sg}}) - \epsilon,
\end{align}
but the optimal representations for each inequality occupy conflicting subspaces---i.e., the gradients $\nabla_\theta$ of the two mutual information objectives are negatively correlated on the melody-encoding subspace.
\end{definition}

\subsection{Full Proof of Proposition~\ref{prop:factorize}}
\label{app:proof_factorize}

We provide the complete proof that factorized conditioning reduces the gradient conflict in joint training.

\begin{proof}
Let $\theta = (\theta_{\text{shared}}, \theta_{\text{cnt}}, \theta_{\text{mel}}, \theta_{\text{tmb}})$ denote the full parameter set, where $\theta_{\text{shared}}$ are the DiT backbone parameters and $\theta_{\text{cnt}}, \theta_{\text{mel}}, \theta_{\text{tmb}}$ are the parameters of the content, melody, and timbre encoders respectively.

\textbf{Step 1: Gradient decomposition.}
The total training gradient is:
\begin{equation}
    \nabla_\theta \mathcal{L} = \alpha \cdot \nabla_\theta \mathcal{L}_{\text{sp}} + (1-\alpha) \cdot \nabla_\theta \mathcal{L}_{\text{sg}},
\end{equation}
where $\alpha$ is the effective mixing ratio of speech vs.\ singing samples.

\textbf{Step 2: Conflict in monolithic conditioning.}
Under a monolithic encoder $f_\theta^{\text{mono}}(\mathbf{c})$ that maps all conditioning signals through shared parameters, consider the gradient on the melody-related subspace $\theta_{\text{mel}}$:
\begin{align}
    \nabla_{\theta_{\text{mel}}} \mathcal{L}_{\text{sp}} &: \text{pushes toward } g_{\theta_{\text{mel}}}(\cdot) \approx \text{const} \quad \text{(melody is irrelevant for speech)}, \\
    \nabla_{\theta_{\text{mel}}} \mathcal{L}_{\text{sg}} &: \text{pushes toward } g_{\theta_{\text{mel}}}(\cdot) \text{ faithfully encoding MIDI notes}.
\end{align}
The cosine similarity between these gradient directions is:
\begin{equation}
    \cos\bigl(\nabla_{\theta_{\text{mel}}} \mathcal{L}_{\text{sp}},\, \nabla_{\theta_{\text{mel}}} \mathcal{L}_{\text{sg}}\bigr) < 0,
\end{equation}
since one drives the parameters toward a constant function while the other drives them toward an informative function.

\textbf{Step 3: Factorization eliminates the conflict.}
In the factorized architecture, the speech loss $\mathcal{L}_{\text{sp}}$ is a function of $(\theta_{\text{shared}}, \theta_{\text{cnt}}, \mathbf{e}_\varnothing, \theta_{\text{tmb}})$ but \emph{not} of $\theta_{\text{mel}}$, because the melody encoder is bypassed and replaced by the null token $\mathbf{e}_\varnothing$.
Therefore:
\begin{equation}
    \nabla_{\theta_{\text{mel}}} \mathcal{L}_{\text{sp}} = \mathbf{0}.
\end{equation}
The melody encoder receives gradients only from singing samples, eliminating the negative gradient correlation entirely.
We define the eliminated conflict magnitude as
\begin{equation}
\Delta
=
-\alpha(1-\alpha)
\left\langle
\nabla_{\theta_{\mathrm{mel}}} L_{\mathrm{sp}}^{\mathrm{mono}},
\nabla_{\theta_{\mathrm{mel}}} L_{\mathrm{sg}}^{\mathrm{mono}}
\right\rangle .
\end{equation}
When the monolithic encoder exhibits negative gradient correlation on the melody-related subspace, the inner product is negative and therefore $\Delta > 0$.
Thus, factorization removes a destructive cross-task interaction from the joint optimization.

\end{proof}

\subsection{Full Proof of Proposition~\ref{prop:null} (Null-Token Optimality)}
\label{app:proof_null}

\begin{proof}
Consider the joint training loss over both speech and singing samples:
\begin{equation}
    \mathcal{L} = \underbrace{\mathbb{E}_{\text{sp}} \bigl[ \| v_\theta(x_t, t, c_{\text{cnt}}, \mathbf{e}_\varnothing, c_{\text{tmb}}, \text{sp}) - u_t \|^2 \bigr]}_{\mathcal{L}_{\text{sp}}} + \underbrace{\mathbb{E}_{\text{sg}} \bigl[ \| v_\theta(x_t, t, c_{\text{cnt}}, c_{\text{mel}}, c_{\text{tmb}}, \text{sg}) - u_t \|^2 \bigr]}_{\mathcal{L}_{\text{sg}}},
\end{equation}
where $u_t = x_1 - x_0$ is the target velocity.
$\mathcal{L}_{\text{sp}}$ depends on $\mathbf{e}_\varnothing$ but not on the melody encoder $g_\phi$; $\mathcal{L}_{\text{sg}}$ depends on $g_\phi$ but not on $\mathbf{e}_\varnothing$.
Therefore $\mathbf{e}_\varnothing$ and $g_\phi$ are optimized by non-overlapping gradient signals, ensuring no interference.
At convergence, the optimal vector field satisfies $v_\theta^{*}(x_t, t, c_{\text{cnt}}, \mathbf{e}_\varnothing, c_{\text{tmb}}, \text{sp}) = \mathbb{E}[u_t \mid x_t, c_{\text{cnt}}, c_{\text{tmb}}, c_{\text{tsk}}\!=\!\text{sp}]$, which is exactly the melody-marginalized conditional expectation of the target velocity---i.e., the Bayes-optimal velocity field for speech.
Since $\mathbf{e}_\varnothing$ is jointly optimized with the backbone $\theta$, it converges to the value that makes this equality hold.
\end{proof}

\begin{remark}[Comparison with fixed-zero baseline]
\label{rem:zero}
A na\"ive alternative is $\mathbf{e}_\varnothing = \mathbf{0}$.
While this removes melody information, it constrains the model to a specific operating point that may not be optimal.
If the melody encoder $g_\phi$ has a non-zero bias, $g_\phi(\cdot)$ evaluated at the data mean may not equal $\mathbf{0}$, introducing systematic bias.
The learned $\mathbf{e}_\varnothing$ absorbs this bias automatically via its $D$ free parameters.
\end{remark}



\subsection{Null Token Training Dynamics}
\label{app:null_dynamics}

We provide additional analysis of how the null token converges during training.

\begin{lemma}[Null token convergence]
\label{lem:null_converge}
Under standard SGD with learning rate $\eta$ on the speech loss $\mathcal{L}_{\text{sp}}$, the null token $\mathbf{e}_\varnothing$ converges to a fixed point $\mathbf{e}_\varnothing^{*}$ satisfying:
\begin{equation}
    \mathbf{e}_\varnothing^{*} = \mathbf{e}_\varnothing^{*} - \eta \cdot \mathbb{E}_{\text{sp}} \left[ \frac{\partial}{\partial \mathbf{e}_\varnothing} \bigl\| v_\theta(x_t, t, c_{\text{cnt}}, \mathbf{e}_\varnothing, c_{\text{tmb}}, \text{sp}) - u_t \bigr\|^2 \right].
\end{equation}
This fixed point satisfies the stationarity condition:
\begin{equation}
    \mathbb{E}_{\text{sp}} \left[ \bigl( v_\theta(x_t, t, c_{\text{cnt}}, \mathbf{e}_\varnothing^{*}, c_{\text{tmb}}, \text{sp}) - u_t \bigr) \cdot \frac{\partial v_\theta}{\partial \mathbf{e}_\varnothing}\bigg|_{\mathbf{e}_\varnothing = \mathbf{e}_\varnothing^{*}} \right] = \mathbf{0}.
\end{equation}
\end{lemma}

This stationarity condition implies that, at the fixed point, the residual error $(v_\theta - u_t)$ is orthogonal to the sensitivity of $v_\theta$ with respect to $\mathbf{e}_\varnothing$.
Intuitively, the null token has converged to a point where no further adjustment can reduce the speech loss---it has found the optimal ``melody-absent'' representation.

\subsection{Information-Theoretic Bound on Joint Training}
\label{app:info_theory}

We derive an information-theoretic bound showing when joint training can improve over separate training.

\begin{proposition}[Positive transfer bound]
\label{prop:transfer}
Let $\mathcal{L}_{\text{sp}}^{\text{joint}}$ and $\mathcal{L}_{\text{sp}}^{\text{sep}}$ denote the speech loss under joint and separate training, respectively.
Under the factorized conditioning scheme, the performance gap is bounded by:
\begin{equation}
    \mathcal{L}_{\text{sp}}^{\text{sep}} - \mathcal{L}_{\text{sp}}^{\text{joint}} \leq I_{\theta_{\text{shared}}}(\text{singing features}; \text{speech generation}),
\end{equation}
where $I_{\theta_{\text{shared}}}$ measures the mutual information between the representations learned from singing data and their utility for speech generation, mediated through the shared backbone parameters $\theta_{\text{shared}}$.
\end{proposition}

\begin{proof}[Proof sketch]
Joint training exposes $\theta_{\text{shared}}$ to both speech and singing data.
The shared backbone can learn representations that capture general vocal production mechanisms (e.g., formant transitions, breathiness, vibrato) from the combined data.
By the data processing inequality, these shared representations can only help (or be neutral) for each individual task, provided the factorized conditioning prevents destructive interference.
The bound follows from bounding the ``useful information'' that singing data provides for speech generation through the shared backbone.
The factorization ensures that the melody-specific information from singing data does not interfere with speech generation (since it flows through $\theta_{\text{mel}}$, not $\theta_{\text{shared}}$), while the shared vocal representations (timbre, formants, articulation dynamics) can transfer positively.
\end{proof}

\begin{remark}
This bound is non-trivial only when $I_{\theta_{\text{shared}}} > 0$, which occurs when there exist common vocal representations useful for both tasks.
In practice, this is likely the case: both speech and singing share articulatory mechanisms, phonetic structures, and speaker-dependent characteristics.
The factorized conditioning ensures that the \emph{conflicting} aspects (explicit vs.\ implicit prosody) are handled separately, while the \emph{complementary} aspects flow through the shared backbone.
\end{remark}

\section{Architecture Details}
\label{app:arch}

\paragraph{DiT Configuration.}
Table~\ref{tab:arch_details} provides the complete architectural specification.

\begin{table}[ht]
\centering
\caption{Detailed architecture configuration of \ours{}.}
\label{tab:arch_details}
\small
\begin{tabular}{ll}
\toprule
\textbf{Component} & \textbf{Specification} \\
\midrule
DiT Layers & 24 \\
Hidden Dimension & 1024 \\
Attention Heads & 16 \\
FFN Expansion & 4$\times$ \\
Total Parameters & 329M \\
\midrule
VAE & Song Bloom VAE~\cite{songbloomvae} \\
VAE Latent Dim & 48 \\
VAE Frame Rate & 25 Hz \\
Input Sample Rate & 48 kHz \\
\midrule
Text Encoder & ConvNeXt (4 layers, kernel=7) \\
Melody Encoder & Conformer (6 layers, MIDI notes) \\
Positional Encoding & RoPE \\
Attention & FlashAttention-2 \\
Normalization & AdaLN (timestep + task) \\
\midrule
Training Epochs & 100 \\
Batch Size & 32{,}000 frames/GPU \\
GPUs & 8$\times$ A800 80GB \\
Optimizer & AdamW \\
Peak LR & $7.5 \times 10^{-5}$ \\
Warmup Steps & 2{,}000 \\
CFG Dropout & 0.1 (text, melody, timbre) \\
\midrule
ODE Solver & Euler \\
Inference Steps & 32 \\
Text Guidance Scale & 5.0 \\
Audio/Melody Guidance & 1.0 \\
\bottomrule
\end{tabular}
\end{table}

\section{UniSinging-Eval Dataset Details}
\label{app:dataset}

To construct a singing test set that (i) covers commonly encountered style types in mainstream popular songs, (ii) remains distinguishable across styles, and (iii) is suitable for reliable re-singing and annotation, we select 12 representative styles from mainstream pop music. Our goal is not to exhaustively enumerate all commercial music genres, but rather to cover a broad yet representative subset of mainstream vocal styles.

The selected styles include long-established, high-level categories that are consistently present in the popular music industry and chart ecosystems---\textbf{Pop, Rock, Jazz, Hip-Hop/Rap, Country,} and \textbf{R\&B}---as well as style-level categories with clear stylistic identities and practical recognizability in contemporary pop contexts---\textbf{City Pop, Neo-Soul, Bossa Nova, Afrobeats, Funk,} and \textbf{Pop Electronic}. The former set is supported by widely adopted genre taxonomies and industry distribution systems such as those used by GRAMMY and Billboard, while the latter aligns with existing music scholarship on topics including style revival (e.g., City Pop), stylistic origins and historical development (e.g., Bossa Nova, Funk), cross-national circulation (e.g., Afrobeats), groove-related characteristics (e.g., Neo-Soul, Funk), and streaming-era style identities (e.g., Pop Electronic)~\cite{genremusic}.

Our style selection is motivated by two considerations. First, these styles span a range of vocal rhythm organizations, melodic tendencies, articulation patterns, and expressive characteristics in contemporary popular singing. Several styles exhibit clear contrasts in vocal delivery and prosodic structure, making them suitable for evaluating generation quality and style adaptability under different vocal styles. For example, \textbf{Rock} singing is often powerful and may include rough or screamed vocals, whereas \textbf{Bossa Nova} is typically soft and can approach a whisper-like delivery; \textbf{Country} melodies tend to be relatively straight and syllable-centered, while \textbf{R\&B} singing often involves extensive melismatic runs and improvisational ornamentation. Second, since UniSinging-Eval is built from human re-singing recordings, we prioritize styles that can be reproduced more stably under controlled recording conditions and that facilitate consistent style annotation and MIDI refinement. This is a deliberate dataset design choice intended to reduce performance variations that are unrelated to the modeling target, thereby improving comparability across samples and the reliability of evaluation results.

To improve annotation consistency and musical interpretability, we organized a team of professional composers to perform a unified manual correction of the MIDI annotations in the test set. The pre-correction MIDI information, besides containing errors, often attempted to preserve nearly all musical parameters present in the final vocal performance, including various instantaneous pitch fluctuations. Our refinement goal is to recover the \emph{composition-intent} level musical parameters as much as possible, so that models can learn the style-dependent ``secondary creation'' (expressive performance variation) on top of this symbolic scaffold.

For portamento, passing tones, and certain ornamental pitch variations, we retain them only when they form a perceptible melodic event; purely transitional pitch movements are not recorded (see Fig.~\ref{fig:midi_adjust}). For rap-like segments, if the segment exhibits a clear and perceptible pitch, we annotate both pitch and rhythm; if the vocal delivery is closer to speech and lacks a stable pitch center, we retain only rhythmic and duration information and do not force pitch assignment. When some segments show global intonation shifts, pitch drift, or fall between adjacent scale degrees, we preferentially map them to the nearest musically self-consistent target pitch structure, reducing the interference of tuning deviations on symbolic melody annotation. For segments with unstable pitch, breathy phonation, or weak pitch evidence, we add pitch labels only when melodic intent can be reliably inferred from local tonal relations, repeated phrases, or surrounding melodic context.

\begin{table}[ht]
\centering
\caption{Configuration and statistics of the UniSinging-Eval dataset.}
\label{tab:dataset}
\small
\begin{tabular}{l|c|c|c|c}
\toprule
\textbf{Setting / Metric} & \textbf{Original} & \textbf{Level 1 (Minor)} & \textbf{Level 2 (Major)} & \textbf{Level 3 (Elastic)} \\
\midrule
\multicolumn{5}{l}{\textit{Dataset Attributes (Common)}} \\
\midrule
\multirow{2}{*}{Styles} & \multicolumn{4}{c}{Pop, Hip-Hop/Rap, Pop Electronic, Rock, Country,} \\
 & \multicolumn{4}{c}{Citypop, Bossa Nova, R\&B, Afrobeats, Jazz, Funk, Neo-Soul} \\
No.\ of Songs & \multicolumn{4}{c}{60} \\
No.\ of Samples & \multicolumn{4}{c}{900} \\
Total Duration (Songs) & \multicolumn{4}{c}{2 h} \\
Total Duration (Samples) & \multicolumn{4}{c}{1.5 h} \\
\midrule
\multicolumn{5}{l}{\textit{Rebuild Constraints}} \\
\midrule
Lyrical Alteration   & None     & Partial (Subset) & Extensive Rewrite & Maximal \\
Word Count Variation  & 0        & 0                & 0                 & $\pm$5 words \\
Theme \& Imagery      & Original & Preserved        & Completely Changed& Arbitrary \\
Vocal \& Content      & Original & Total/Rhyme Match& Complex Vocab     & Filler/Transition Words \\
\bottomrule
\end{tabular}
\end{table}

\begin{figure*}[t]
    \centering
    \includegraphics[width=0.9\textwidth]{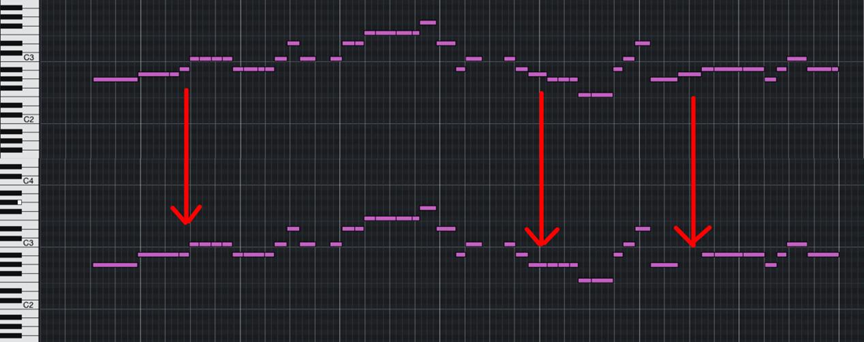}
    \caption{Manual MIDI refinement. The top piano-roll shows raw pitch annotations with instantaneous fluctuations, while the bottom shows corrected MIDI notes. Transitional pitch movements such as glides and micro-fluctuations, highlighted by red arrows, are removed, and notes are mapped to musically consistent target pitches to better reflect composition-level melodic intent.}
    \label{fig:midi_adjust}
\end{figure*}

\end{document}